\documentclass[a4paper,11pt]{article}
\usepackage[margin=1in]{geometry}
\usepackage{amsmath,amssymb}
\usepackage{hyperref}
\hypersetup{
	colorlinks=true, 
	linktoc=all,     
	linkcolor=blue,
	citecolor=red,  
}

\usepackage{graphicx}
\usepackage{amsmath}
\usepackage{amsthm}
\usepackage{booktabs}
\usepackage{algorithm}
\usepackage{algorithmic}
\usepackage[switch]{lineno}

\usepackage{mathtools}

\usepackage{amsthm}
\usepackage{caption} 
\usepackage{thmtools}
\usepackage{thm-restate}
\usepackage{todonotes}
\newtheorem{remark}{Remark}
\usepackage{newfloat}
\usepackage{listings}
\usepackage{tcolorbox}
\usepackage{multirow}
\usepackage{threeparttable}
\usepackage{framed}
\usepackage{xcolor}
\usepackage{xspace}
\usepackage{caption}
\usepackage{soul}

\author{
	Shiri Alouf-Heffetz\thanks{Ben-Gurion University, Beer Sheva, Israel. \texttt{shirihe@post.bgu.ac.il}}
	\and
	Tanmay Inamdar\thanks{
		Indian Institute of Technology Jodhpur, Jodhpur, India. \texttt{taninamdar@gmail.com}}
	\and
	Pallavi Jain\thanks{
		Indian Institute of Technology Jodhpur, Jodhpur, India. \texttt{pallavijain.t.cms@gmail.com}}
	\and
		Yash More\footnote{Indian Institute of Technology Gandhinagar, Gandhinagar, India. \texttt{yash.mh@iitgn.ac.in}}
	\and
	Nimrod Talmon\footnote{Ben-Gurion University, Beer Sheva, Israel. \texttt{talmonn@bgu.ac.il}
	\\Authors are specified in the alphabetical order of last names. 
	\\This research was done when T.~Inamdar was affiliated with University of Bergen and he acknowledges support from the European Research Council (ERC) under the European Union’s Horizon 2020 research and innovation programme (grant agreement No. 819416). P.~Jain acknowledges support from SERB-SUPRA grant number S/SERB/PJ/20220047 and IITJ Seed Grant grant I/SEED/PJ/20210119.} 
}

\date{}

\title{Controlling Delegations in Liquid Democracy}

\newcommand{\lr}[1]{\left(#1\right)}
\newcommand{\LR}[1]{\left\{#1\right\}}

\newcommand{\calR}{\mathcal{R}}
\newcommand{\calF}{\mathcal{F}}
\newcommand{\cost}{cost}
\newcommand{\ccrafull}{{\textsc{Constructive Control by Redirecting Arcs}}\xspace}
\newcommand{\poly}{{\textsf{P}}}
\newcommand{\npc}{{\textsf{NP}-complete}}
\newcommand{\nph}{{\textsf{NP}-hard}}
\newcommand{\fpt}{{\textsf{FPT}}}
\newcommand{\Wtwo}{{\textsf{W}[2]-hard}\xspace}
\newcommand{\xp}{{\textsf{XP}}}
\newcommand{\deleg}{\#{\texttt{delegations}}}
\newcommand{\app}{\#{\texttt{approvals}}}
\newcommand{\redirect}{\#{\texttt{redirections}}}

\newcommand{\ccra}{{\textsc{CCRA}}\xspace}
\newcommand{\vc}{{\textsc{Vertex Cover}}\xspace}

\newcommand{\defprob}[3]{
\begin{tcolorbox}[colback=gray!5!white,colframe=gray!75!black]
  \begin{tabular*}{\textwidth}{@{\extracolsep{\fill}}lr} #1   \\ \end{tabular*}
  {\bf{Input:}} #2  \\
  {\bf{Question:}} #3
  \end{tcolorbox}
}

\newtheorem{example}{Example}
\newtheorem{theorem}{Theorem}
\newtheorem{corollary}{Corollary}
\newtheorem{lemma}{Lemma}
\newtheorem{claim}{Claim}

\begin{document}


\maketitle 

\begin{abstract}
	In liquid democracy, agents can either vote directly or delegate their vote to a different agent of their choice. This results in a power structure in which certain agents possess more voting weight than others. As a result, it opens up certain possibilities of vote manipulation, including control and bribery, that do not exist in standard voting scenarios of direct democracy. Here we formalize a certain kind of election control -- in which an external agent may change certain delegation arcs -- and study the computational complexity of the corresponding combinatorial problem.
\end{abstract}

\section{Introduction}

Liquid democracy is an innovative approach to democratic governance that can be thought of as a middle-point between direct and representative democracy.  It has attracted significant attention both from practitioners and from academics, mainly as it offers greater flexibility to voters:
  essentially, voters participating in liquid democracy can not only choose how to fill their ballot, but they can choose instead of filling their ballot to delegate their vote to another voter of their choice.

More elaborately, unlike traditional models where agents either directly participate in the decision-making process by explicitly casting their vote -- as in direct democracy -- or elect representatives to make decisions on their behalf -- as in representative democracy, liquid democracy offers a hybrid system that allows for both direct and indirect participation. At its core, liquid democracy enables individuals to delegate their voting power to trusted proxies while retaining the option to vote directly on specific issues. This fluidity of participation holds the potential to enhance democratic engagement, facilitate more informed decision-making, and address the limitations of existing democratic systems.

One of the key advantages of liquid democracy indeed lies in the added flexibility it provides to voters. By allowing individuals to delegate their voting power to trusted proxies, liquid democracy empowers citizens to actively participate in decision-making even when they are unable to devote extensive time and effort to every issue. This flexibility enables individuals to delegate their votes to experts or representatives they trust on specific subjects, while still retaining the ability to directly vote on matters that are of particular importance to them. Importantly, this has the potential of increasing the quality of the decision making process.

Like any democratic system, however, liquid democracy is not immune to potential challenges, including the risk of voter manipulation. Given the fluid nature of delegation, there is the natural concern that influential or malicious actors could exert undue influence by strategically delegating votes or exploiting the trust placed in them as proxies. Such manipulation could undermine the principles of fairness and equality that are fundamental to the democratic process. The issue of the vulnerability of liquid democracy to certain different forms of manipulation has yet to receive sufficient attention from the research community.

Here we consider several forms of manipulation in the context of liquid democracy. In particular, we concentrate on the computational complexity of successfully conducting these manipulations. In particular, we are interested in the situation in which some external agent tries to rig the result of the election by redirecting few delegation arcs.

We define a corresponding combinatorial problem and study its computational complexity. We observe cases that allow for polynomial-time algorithms, cases that allow for \xp\ algorithms, parameterized (approximation) algorithms, as well as cases for which no efficient algorithms exist (assuming certain complexity theoretic hypothesis, such as $P \neq NP$).

We wish to stress that liquid democracy is currently utilized in practical settings for decisions of varying significance, from political realms like the German Pirate Party to financial contexts in various blockchain ventures. Its appeal lies both in its potential to enhance democratic participation as well as its promise to improve the quality of the decision making itself. Thus, we argue that, recognizing the limitations and vulnerability of liquid democracy to external forms of influence -- such as the ones we study in the current paper -- is essential for a better understanding and usage of liquid democracy.

\section{Related Work}

By now, liquid democracy has attracted significant attention from the research community. In particular, it is studied from a political science perspective~\cite{blum2016liquid,paulin2020overview}; from an algorithmic point of view~\cite{kahng2021liquid,dey2021parameterized,caragiannis2019contribution}; from a game-theoretic point of view~\cite{zhang2021power,bloembergen2019rational,escoffier2019convergence}; and from an agenda aiming at pushing its boundaries by increasing the expressiveness power that is granted to voters~\cite{brill2018pairwise,jain2022preserving,DBLP:conf/atal/KoppeKST24}.

Here we study election control in liquid democracy; in this context, we mention the extensive survey on control and bribery~\cite{faliszewski2016control}. Indeed, due to their importance, the problems of control and bribery are studied thoroughly~\cite{bartholdi1992hard,procaccia2007multi,bredereck2016complexity}.

To the best of our knowledge, however, the topic of election control and bribery in the context of liquid democracy was not yet studied. We do mention the work of \cite{zhang2021power} and its follow up paper~\cite{d2021computation}, in which agents may behave strategically. Note, however, that we are interested not in vote manipulation (by the community members themselves) but in election control (by an external agent).
Thus, from this point of view, our work extends the vast literature on control and bribery -- that is currently applied to settings such as single-winner elections and multiwinner elections -- to the setting of liquid democracy. And, as the most distinguishing feature of liquid democracy is its usage of vote delegations, we study the natural control action of altering the delegations themselves.

Another related line of research focuses on opinion diffusion in social networks. Following vote delegation chains in liquid democracy mirrors the process of tracking preferences in opinion diffusion. Notably, \cite{faliszewski2022opinion,abouei2020multi} provided insights into opinion propagation using a voting rule similar to our approach for resolving multi-delegations.

Finally, note that the type of liquid democracy we consider here is not only the standard type in which a voter may delegate their vote to a single other voter but also a version in which a voter may delegate to several other voters, in which case the votes of this set of voters is aggregated to compute the vote of the delegating voter. In this context, we are closer to more advanced types of liquid democracy~\cite{colley2022unravelling,golz2021fluid}.





\section{Control by Redirecting Arcs}

Next, we first position our work within the landscape of control problems for liquid democracy; and then describe our formal model.

\subsection{Liquid Democracy and Election Control}

Generally speaking, when considering election control in liquid democracy, there are several orthogonal factors to take into account:
\begin{itemize}

\item
The ballot type: e.g., whether ballots are binary, plurality, approval, ordinal, cumulative.

\item
The delegations type: e.g., whether delegations are transitive, coarse-grained or fine-grained~\cite{brill2018pairwise,jain2022preserving}, whether several delegations are possible for a single agent.

\item
The actions at the disposal of the controlling agent: e.g., whether the controlling agent can add edges, remove edges, redirect edges, change the ballots of some voters.

\item
The goal of the controlling agent: e.g., whether the goal is to make some predefined preferred candidate win the election.

\end{itemize}

Following the taxonomy above, it is worthwhile to note that here we consider: approval ballots; transitive coarse-grained ballots with possibly more than one delegate to each agent; a controlling agent that can only redirect edges and whose goal is to make some predefined preferred candidate a winner of the election after the controlling actions and the unraveling of the (partially modified) delegations. 

\subsection{Formal Model}

We describe our formal model.
Our formal model contains the following ingredients:
\begin{itemize}
\item A set of \emph{voters} $V = \{v_1, \ldots, v_n\}$.

\item A \emph{delegation graph} $G = (V, E)$, which is a directed graph where the voters are the vertices.
If the delegation graph has an out-degree of at most one,
then we sometimes refer to the situation as a ``single-delegation'' setting;
while, in general, it is a ``multi-delegation'' setting.
A voter $v_i$ with out-degree $0$ is said to be an \emph{active voter},
while a voter $v_i$ with out-degree strictly positive is said to be a \emph{passive voter}.

\item 
A cost function $\cost \colon E(G)\rightarrow \mathbb{N}$. The cost of an arc in the delegation graph $G$ is the cost of redirecting this arc.
\item
We consider approval elections, so there is a (global) set of candidates
$C = \{c_1, \ldots, c_m\}$ and each active voter $v_i$ corresponds to a ballot $b_i \subseteq C$.

\item
An \emph{unraveling function} $\mathcal{R}$ that takes several ballots and returns a ballot; it is used in the following way: the ballot of a passive voter delegation to active voters with ballots $b_1, \ldots, b_z$ is unravelled to the ballot $\mathcal{R}(b_1, \ldots, b_t)$, as explained next.) Note that, throughout the paper we consider several rules $\mathcal{R}$ as the unravelling function as described below. 
In particular,
we consider the following rules:
\begin{itemize}

\item
{\bf The union function:} here, $\mathcal{R}$ returns the union of the ballots it gets as input: $\mathcal{R}(b_1, \ldots, b_t) = \cup_{z \in [t]} b_z$.

\item
{\bf Approval function:} an approval function takes a ballot and returns the set of candidates approved by the maximum number of voters.

\item
{\bf GreedyMRC function}: Under GreedyMRC, we return a set of candidates obtained by the following method: we start with an empty set and perform a sequence of iterations, where in each iteration we (a) add to the current set, a set of candidates $S$ approved by the largest number of voters (all the candidates in $S$ are approved by the same number of voters) and (b) remove the voters that approve candidates in $S$ from further consideration.

\end{itemize}

Note that the choice of unraveling function is only relevant for multiple delegation.

\item
Given a delegation graph $G$ and an unraveling function~$\mathcal{R}$, we can define the \emph{unraveled vote} of each passive voter by following its delegations transitively; in particular, a passive voter $v_i$ with out-arcs to $u_1, u_2, u_3$ -- say, all of which are active voters -- will be assigned the unravelled ballot $b_i = \mathcal{R}(u_1, u_2, u_3)$; this happens transitively (i.e., at each level).
For simplicity, assume that there are no cycles.From real-world applications of liquid democracy, in particular in the context of the LiquidFeedback platform, we know that usually the number of cycles is rather limited in practice.
Let $b_i$ be the ballot of $v_i$ if $v_i$ is an active voter and its unravelled vote if $v_i$ is a passive voter.

\item
A \emph{voting rule} $\mathcal{W}$ that takes a set of (possibly some of which are unravelled) ballots and returns a candidate $c^\star \in C$ as the winner of the election. 
Note that, throughout the paper we consider \textbf{$\mathcal{W}$ as the Approval voting rule}, which returns the name of the candidate with the highest number of votes; we discuss other options in the Outlook section towards the end of the paper.

\item
An external agent -- referred to as \emph{the controller} -- who has a predefined budget of $k$ and who can change a set of arcs, say $S$, whose total cost is at most $k$ 
and whose goal is to make some predefined preferred candidate $c^\star$ a winner of the election where we first redirect the arcs in~$S$, 
then unravel all delegations using $\mathcal{R}$, then apply $\mathcal{W}$ on the (direct and unravelled) ballots, and then see whether $c^\star$ is in the winning committee $w$.

By redirecting an arc we mean taking an arc $(u, v)$, removing it from the graph, and inserting a different arc $(u, v')$.

\end{itemize}

So, formally, the problem we are considering in this paper can be defined as follows:

\defprob{\ccrafull \\ (\ccra)}{A delegation graph $G=(V,E)$ with costs function $cost$ with approval ballots over a set of candidates $C$, an unraveling function $\mathcal{R}$, a voting rule $\mathcal{W}$, a preferred candidate $c^\star$, and a budget $k$.}{Does there exist a delegation graph $G'$ that is achieved from $G$ by performing redirections within the budget $k$, such that, after unraveling all of the delegations of $G'$ using $\mathcal{R}$, $c^\star$ is the unique winner of the resulting election using $\mathcal{W}$?}

We denote the number of delegation as \deleg\  and the number of approvals as \app.

The following example illustrates our model. 
\begin{example} Consider the delegation graph in Figure~\ref{fig:example} on the voter set $\{v_1,v_2,v_3,v_4,v_5,v_6\}$ and the candidate set $\{c_1,c_2,c_3,c^\star\}$, where $c^\star$ is the preferred candidate. Let the cost of every arc be $1$. Consider union function as the unraveling function $\mathcal{R}$ and approval voting as the voting rule $\mathcal{W}$. In Figure~\ref{fig:example}(a), since $v_2$ delegates to $v_3$ and $v_4$, the vote of $v_2$ is $\{c_1,c_2,c^\star\}$. Similarly, the vote of $v_5$ is $\{c_1,c_2\}$ and the vote of $v_1$ is $\{c_1,c_2,c_3,c^\star\}$. Since $c_1$ is approved by the maximum number of voters, $c_1$ is the winner in the delegation graph in Figure~\ref{fig:example}(a). We redirect the arc $(v_5,v_3)$ to $(v_5,v_6)$ as depicted in Figure~\ref{fig:example}(b). Now, $c^\star$ is approved by $5$ voters (all but $v_3$ approves $c^\star$) and no other candidate is approved by $5$ voters. Thus, $c^\star$ is the unique winner in the delegation graph in Figure~\ref{fig:example}(b).
\begin{figure}
    \centering
    \includegraphics[scale=0.4]{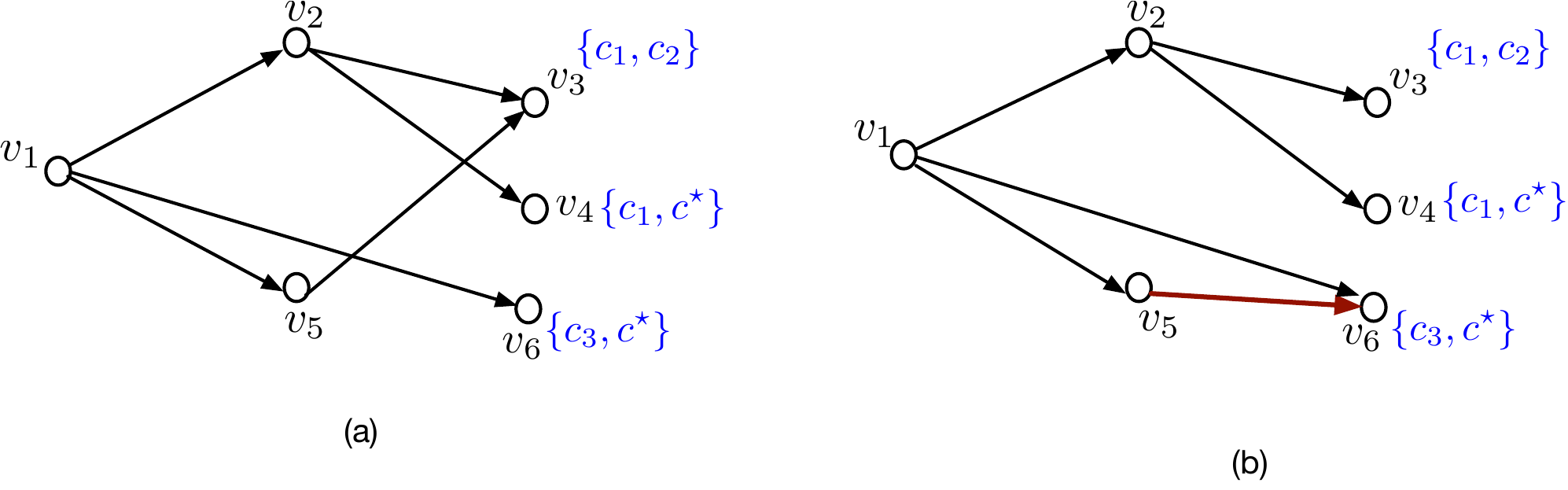}
    \caption{(a) Input delegation graph, (b) Delegation graph after redirecting arc $(v_5,v_3)$ to $(v_5,v_6)$ }
    \label{fig:example}
\end{figure}
\end{example}

Table~\ref{tab:classical_complexity} and Table~\ref{tab:parameterized_complexity} summarise the classical and parameterized complexity of our problem, respectively.

\begin{table*}[ht]
\caption{Computational Complexity of \ccra}
\centering
\footnotesize
\begin{tabular}{c c | c c}
 \hline
 \multicolumn{2}{c}{\bf \#delegations=1} & \multicolumn{2}{c}{\bf \#delegations $\ge$ 1} \\
 \hline
 \#approvals = 1 & \#approvals $\ge$ 1 & \#approvals = 1 & \#approvals $\ge$ 1 \\
 \hline
 \poly (Thm~\ref{thm:ooly-single-deleg}) & \npc (Thm~\ref{thm:nph-single-deleg}) & \npc (Thm~\ref{thm:nph-general}) & \npc (Thm~\ref{thm:nph-general}) \\
 \hline
\end{tabular}
\vspace{0.1in}
\label{tab:classical_complexity}
\end{table*}

\begin{table}[ht]
\begin{threeparttable}[b]
\caption{Parameterized Complexity of \ccra. Here, $\delta^+(G)$ is the maximum out-degree of a vertex in $G$, $\delta^-(G)$ is the maximum in-degree of a vertex in $G$, and $\lambda(G)$ is the maximum length of a path in the delegation graph $G$. ? denote that the complexity is open.}
\centering
 \footnotesize
\begin{tabular}{c| c | c c}
 \hline
  \multirow{2}{*}{Parameters} & {\bf \#delegations=1} & \multicolumn{2}{c}{\bf \#delegations $\ge$ 1} \\
  \cline{2-4}
  & \#approvals $\ge$ 1 & \#approvals = 1 & \#approvals $\ge$ 1 \\
 \hline
 \#voters ($n$) & \fpt (Thm~\ref{thm:fpt-n}) & \fpt (Thm~\ref{thm:fpt-n}) & \fpt (Thm~\ref{thm:fpt-n}) \\
 \#active voters ($t$) & \xp (Thm~\ref{thm:xp-active}),  {\sf FPT-AS}\tnote{a}  (Thm~\ref{thm:fpt-as}) & ? & ? \\
 \#candidates ($m$) & \xp,  {\sf FPT-AS}\tnote{b} (Cor~\ref{cor:xp-m}) & ? & ? \\
 $\delta^+(G)+\delta^-(G)+\app$ & para\nph (Thm~\ref{thm:nph-single-deleg}) & para\nph (Thm~\ref{thm:nph-general})& para\nph (Thm~\ref{thm:nph-general}) \\
 $\delta^+(G)+\lambda(G)+\app$ & para\nph (Thm~\ref{thm:nph-single-deleg}) & para\nph (Thm~\ref{thm:nph-general})& para\nph (Thm~\ref{thm:nph-general}) \\
 $k+\redirect$ & $\mathsf{W}[2]$-hard\tnote{c}      (Thm~\ref{thm:nph-single-deleg}) & $\mathsf{W}[2]$-hard\tnote{d} (Thm~\ref{thm:nph-general}) & $\mathsf{W}[2]$-hard\tnote{d} (Thm~\ref{thm:nph-general}) \\
 \hline
\end{tabular}
\vspace{0.1in}
\label{tab:parameterized_complexity}
 \begin{tablenotes}
       \item[a] The parameter is $t+\epsilon$. 
       \item[b] The parameter is $m+\epsilon$.
       \item[a,b] For both the results, we assume that there is an active voter who either approves only $c^\star$ or any subset of candidates that excludes $c^\star$. 
       \item[d] The result holds even when $\delta^+(G)+\delta^-(G)+\app$ is constant or $\lambda(G)+\app$ is constant. 
       \item[e] The result holds even when $\delta^-(G)+\app$ is constant or $\lambda(G)+\app$ is constant. 
     \end{tablenotes}
\end{threeparttable}
\end{table}

\begin{remark}
Note that $\mathcal{W}$ can be any single-winner voting rule operating on approval ballots. Similarly, $\mathcal{R}$ can be any committee selection rule with variable number of winners (usually referred to as a VNW rule)~\cite{vnw}. As such, it means that our model is general enough to apply any VNW rule to its unravelling procedure and any single-winner voting rule as its winner selection procedure.
\end{remark}


\begin{remark}\label{remark:arbitrarygraph}
In many works on liquid democracy there is an underlying social graph connecting the voters and voters can only delegate to their corresponding friends. Note that, in our setting, a voter can delegate to any other voter; so, put differently, we study the special case in which the underlying social graph is a clique. Since we study a special case, our hardness results hold also for the general case with an arbitrary underlying social graph. 
\end{remark}


\section{Hardness Results}

We begin with our results regarding the computational hardness of the \ccra problem. Both our {\sf NP}-hardness results are due to the polynomial-time reduction from the \vc problem in cubic graphs and the ${\sf W[2]}$-hardness results are due to the same reduction from the {\sc Hitting Set} problem.   

In the \vc problem in cubic graphs, given a cubic graph $G=(V,E)$ (the degree of every vertex is three), and an integer $\tilde{k}$, the goal is to find a set $S\subseteq V(G)$ of size at most $\tilde{k}$ such that at least one of the endpoint of every edge is in $S$. 
The problem is known to be \nph~\cite{DBLP:conf/stoc/GareyJS74}. In the {\sc Hitting Set} problem, given a universe $U$, a family, ${\calF}$, of subsets of $U$, and an integer $\tilde{k}$, the goal is to find a set $S\subseteq U$ such that for every set $F\in {\calF}$, $S\cap {\calF} \neq \emptyset$. Note that {\sc Hitting Set} is a generalisation of the \vc problem, and it is known to be \Wtwo with respect to the parameter $\tilde{k}$~\cite{DBLP:journals/siamcomp/DowneyF95}.

We begin with our first hardness result.
\begin{restatable}{theorem}{hardnessone} \label{thm:nph-general}
    \ccra is \nph\ when $\mathcal{R}$ is union function or approval function or GreedyMRC funtion even when  
    \begin{enumerate}
        \item every vertex in the delegation graph satisfies one of the following conditions:
        \begin{enumerate}
         \item out-degree is at most three, the maximum length of a path in the delegation graph is at most two, and the cost of arcs belong to $\{1,2\}$
         \item out-degree is at most two and in-degree is at most two 
        \end{enumerate}
        \item size of approval set is one
        \item only one voter approves $c^\star$ 
    \end{enumerate}
    Furthermore, it is \Wtwo\ with respect to $k+\redirect$ with all the above constraints except that in $1(a)$ out-degree is not bounded by a constant.
\end{restatable}
\begin{proof}
We focus the proof for the union function. The proof is exactly the same for other two unraveling functions as well since in our gadget the set of candidates obtained by all these unraveling functions is same.
 
    We give a polynomial time reduction from the \vc problem in cubic graphs.
    Let $(G,\tilde{k})$ be an instance of \vc such that $G$ is a cubic graph. We construct an instance of \ccra that satisfies constraints $1(a), 2$, and $3$ in the theorem statement. We will discuss later the required modifications for constraint $1(b)$. The construction is as follows. 
    \begin{itemize}
        \item For every edge $e\in E(G)$, we have a candidate $c_e$. Additionally, we have a candidate $c^\star$ (our preferred candidate).
        \item For every vertex $v\in V(G)$, we have two passive voters $v$ and $v'$, and $v$ delegates to $v'$. The cost of this arc is $1$.
        \item  For every edge $e\in E(G)$, we have an active voter $e$ who approves the candidate $c_e$. If $u$ is an endpoint of $e$, then $u'$ delegates to $e$ and the cost of this arc is $2$. 
        \item  We add a special voter $v^\star$ who approves the candidate $c^\star$.
        \item For every edge $e\in E(G)$, we add a set of dummy voters $D^e=\{d^e_1,\ldots,d^e_{k-5}\}$. Each dummy voter in $D^e$ delegates to the active voter $e$ and the cost of this arc is $2$.
        \item The budget $k$ is $\tilde{k}$.
    \end{itemize}
Let $H$ be the delegation graph that we constructed. Note that in this graph the score of $c^\star$ is $0$ and $k$ for every other candidate. Figure~\ref{fig:thm1-1} illustrates the construction. 

\begin{figure}[ht]
    \centering
    \includegraphics[scale=0.55]{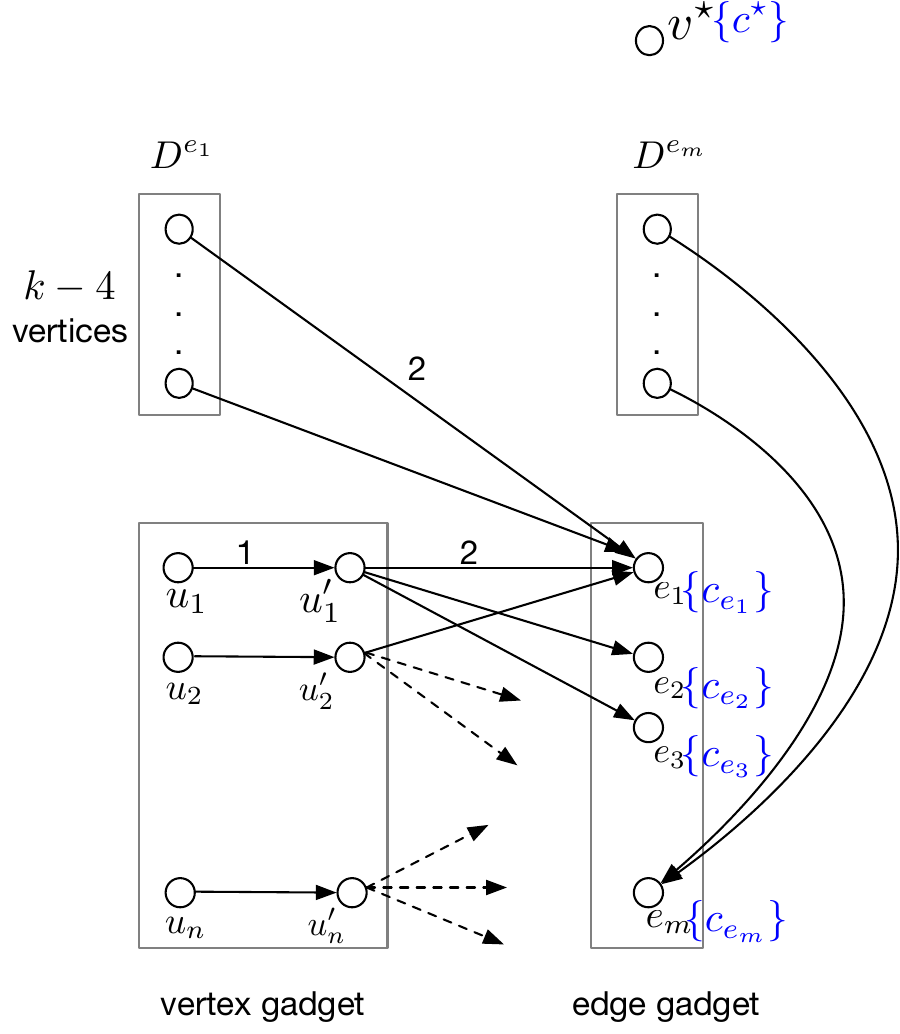}
    \caption{Illustration of {\sf NP}-hardness claimed in Theorem~\ref{thm:nph-general} with constraints 1(a), 2, and 3. The dotted arcs are to illustrate that the out-degree of every $u'$ is three. Here, $e_1=u_1u_2$ is an edge in $G$. An approved candidate by an active voter is written in the {\color{blue}blue} color next to the voter name. The number on the top of an edge is the cost of redirection.}
    \label{fig:thm1-1}
\end{figure}

Next, we prove the correctness. 
 In the forward direction, suppose that $S$ is a solution to $(G,k)$. Without loss of generality, we assume that $|S|=k$. If $v\in S$, then we redirect the arc $(v,v')$ to $(v,v^\star)$. Since the cost of arc $(v,v')$ is $1$ and $|S|= k$, the total cost of redirection is $k$. Now, we prove that $c^\star$ is the unique winner in the delegation graph $H^\star$ obtained after redirections. For every vertex $v\in S$, we have an arc $(v,v^\star)$ in $H^\star$. Furthermore, since $v$ is the source vertex (vertex with in-degree zero) in the delegation graph, the vote of $c^\star$ is $k$. For every edge $e(=uv) \in E(G)$, $u'$ and $v'$ delegates to $e$. Since either $u$ or $v$ is in $S$, the vote of candidate $c_e$ decreases by at least $1$, hence, the score of $c_e$ in $H^\star$ is at most $k-1$ ($k-5$ votes are due to dummy voters). Thus, $c^\star$ is the unique winner in $H$. 

 In the reverse direction, suppose that $H^\star$ is the delegation graph obtained after redirecting arcs within the budget such that $c^\star$ is the unique winner. Let $S$ be the set of arcs that are redirected in $H$ to obtain $H^\star$. We can safely assume that all the arcs are redirected to $v^\star$ because if an arc is redirected to some other voter say $v$, then instead of $v$, redirecting it to $v^\star$ increases the score of $c^\star$ and decrease the score of other voters. We modify $S$ a bit as follows. Suppose that an arc $(d,e)\in S$, and $(u',e) \notin S$, where $u'$ is a voter corresponding to the vertex $u\in V(G)$, which is an endpoint of the edge $e\in E(G)$. Then, delete $(d,e)$ from $S$ and  add $(u',e)$ in $S$. Note that the cost of both the arcs is same, so there the total cost remain same. Redirection of arc $(d,e)$ decreases the score of the candidate $c_e$ by $1$ and increase the score of $c^\star$ by $1$. However, redirection of arc $(u',e)$ decreases the score of the candidate $c_e$ by at least $1$ and increase the score of $c^\star$ by at least $1$ (``at least" because we might have already redirected arc $(u,e)$). Thus, $S$ is still a solution. We can safely assume that if an arc $(u',e)$ is redirected, then $(u,u')$ is not redirected as it only contributes in the cost, but does not increase/decrease the score. 

 We construct a set $S'\subseteq V(G)$ as follows: if an edge incident to $u'$ ($(u,u')$ or $(u',e)$) is in $S$, add the vertex $u\in V(G)$ in $S'$. Next, we prove that $S'$ is a vertex cover of $G$ of size at most $k$. Let $\alpha, \beta, \gamma$ be the number of arcs in $S$, that are incident on dummy voters, arcs of {\emph type} $(u',e)$, and arcs of \emph{type} $(u,u')$. Thus, the cost of redirection is $2\alpha+2\beta+\gamma$ and the score of $c^\star$ is $\alpha+2\beta+\gamma$. 

 \begin{claim}\label{clm:nph}
     By redirecting arcs in $S$, the score of every candidate $c_e$, where $e\in E(G)$, reduces by at least $1$.
 \end{claim}

 \begin{proof}
     Note that in the delegation graph $H$, the score of every candidate $c_e$, where $e\in E(G)$, is $k$ and the score of $c^\star$ is $0$. If there exists a candidate whose score is $k$ in the delegation graph $H^\star$, then the score of $c^\star$ is at least $k+1$. Thus, $\alpha+2\beta+\gamma > k$. Hence, $2\alpha+2\beta+\gamma > k$, a contradiction to the fact that the cost of redirection is at most $k$.
 \end{proof}

Due to Claim~\ref{clm:nph} and the modified $S$, we know that for every voter $e$, either $(u',e)\in S$ or $(u,u')\in S$ (not both), where $u$ is an endpoint of edge $e\in E(G)$. Thus, due to the construction of $S'$, for every edge $e\in E(G)$, at least one of its endpoint is in $S'$. We next claim that the size of $S'$ is at most $k$. For every vertex $u\in V(G)$, either $(u,u')\in S$ or $(u',e)$, not both. Thus, $|S'|=\beta + \gamma$. Since $2\alpha+2\beta+\gamma \leq k$, it follows that $\beta+\gamma \leq k$. Hence, $|S'|\leq k$. 

This completes the proof of Theorem~\ref{thm:nph-general} with constraints 1(a), 2, and 3. 

Next, we mention the required modification to decrease the in-degree and out-degree of the delegation graph. Figure~\ref{fig:thm1-2} illustrates these modifications. Now, every vertex $u\in V(G)$, we add four passive voters $u,u',\hat{u}$, and $\tilde{u}$; $u$ delegates to $u'$ as earlier, and $u'$ delegates to $\hat{u}$ and $\tilde{u}$. For every edge $e\in E(G)$, we have a passive voter $e$ and an active voter $e'$ who votes for the candidate $c_e$. The passive vote $e$ delegates to $e'$. Let $e_1,e_2,e_3$ be the edges incident to the vertex $u$ in $G$. Then, $\hat{u}$ delegates to $e_1$ and $e_2$, and $\tilde{u}$ delegates to $e_3$. In particular, $\hat{u}$ delegates to two passive voters corresponding to two edges incident to $u$ in $G$, and $\tilde{u}$ delegates to the passive voter corresponding to the third edge incident to $u$ in $G$. Note that the score of $c_e$ in $H$ was $k$, where $e\in E(G)$. Thus, to meet the same score, we add $k-8$ dummy passive voters in $D^e$. Let $D^{e}=\{d^e_1,\ldots,d^e_{k-8}\}$. Then, for every $i\in [k-7]$, $d^e_i$ delegates to $d^e_{i+1}$, and $d^e_{k-8}$ delegates to $e'$. The budget is same as earlier, i.e., $k=\tilde{k}$. To ensure that the arc $(e,e')$ is not redirected, we set the cost of this arc as $k+1$, for every $e\in E(G)$. Furthermore, we set the cost of arc $(u,u')$ as $1$ and the remaining arcs as $k+1$. The correctness is same as earlier. Note, that now we can only redirect the arcs of type $(u,u')$ within the budget. 

\begin{figure}[ht]
    \centering
    \includegraphics[width=8cm,height=5.5cm]{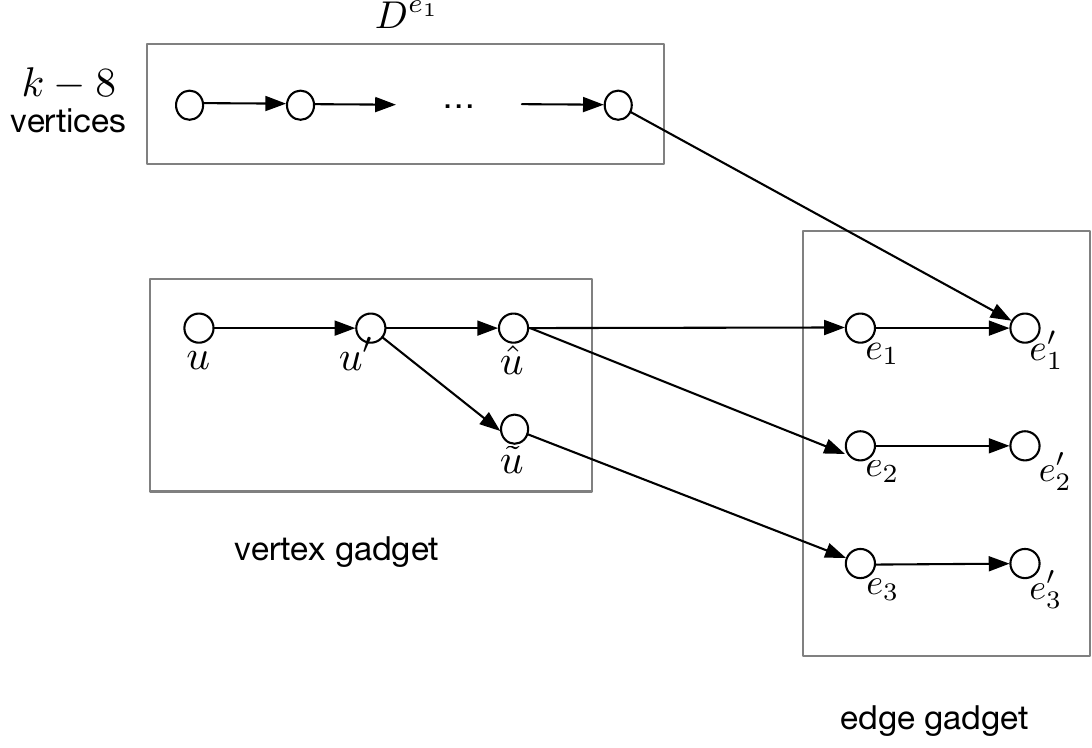}
    \caption{Modifications in Figure~\ref{fig:thm1-1} in the vertex gadget, edge gadget, and dummy vertices for {\sf NP}-hardness of Theorem~\ref{thm:nph-general} with constraint 1(b). Here, $u$ is an endpoint of $e$ in $G$.}
    \label{fig:thm1-2}
\end{figure}
Instead of \vc, if we give the same reduction from the {\sc Hitting Set} problem, we obtain the claimed {\sf W}[2]-hardness. 
\end{proof}

Since the problem is \nph\ even for two delegations, next we study the problem for single delegation. Unfortunately, we again have a negative result. 

\begin{restatable}{theorem}{hardnesstwo}\label{thm:nph-single-deleg}
     \ccra is \nph  \ when ${\calR}$ is union function or approval function or GreedyMRC function, $\deleg = 1$ even when
     \begin{enumerate}
         \item the delegation graph satisfies one of the following conditions: 
         \begin{enumerate}
         \item the maximum length of a path in the delegation graph is at most one and  the cost of all the arcs is $1$
        \item every vertex in the delegation graph has in-degree at most one
         \end{enumerate}
        \item size of approval set is at most three
        \item only one voter approves $c^\star$ 
     \end{enumerate}
     Furthermore, it is \Wtwo\ with respect to $k+\redirect$ even when $\deleg=1$ and condition $1$ and $3$ holds.
\end{restatable}

\begin{proof}
      We focus the proof for the union function. The proof is exactly the same for other two unraveling functions as well since in our gadget the set of candidates obtained by all these unraveling functions is same.
     
We give a polynomial time reduction from the \vc problem in cubic graphs. 
    Let $(G,\tilde{k})$ be an instance of \vc such that $G$ is a cubic graph. The idea is similar to the one used to prove Theorem~\ref{thm:nph-general}. Here, $u$ delegates to $u'$ and $u'$ approves all the candidates corresponding to the edges incident on $u$. Dummy voters in $D^e$ delegate to the voter corresponding to the edge $e$, who approves $c_e$. Figure~\ref{fig:thm1-3} illustrates the construction. 

    Next, we present the construction, in detail.  We first construct an instance of \ccra that satisfies constraints $1(a), 2$, and $3$ in the theorem statement. We will discuss later the required modifications for constraint $1(b)$. The construction is as follows. 

    \begin{itemize}
        \item For every edge $e\in E(G)$, we have a candidate $c_e$. Additionally, we have a candidate $c^\star$ who is our favorite candidate. For every $u\in V(G)$, let $E_u$ be the set of candidates corresponding to the edges incident to $u$ in $G$.
        \item For every vertex $u\in V(G)$, we have a passive voter $u$ and an active voter $u'$. The voter $u$ delegates to $u'$ and the voter $u'$ approves all the candidates in $E_u$. 
        \item  For every edge $e\in E(G)$, we have an active voter $e$ who approves the candidate $c_e$.
        \item For every edge $e\in E(G)$, we add a set of dummy voters $D^e=\{d^e_1,\ldots,d^e_{k-4}\}$. Each dummy voter in $D^e$ delegates to the active voter $e$ and the cost of this arc is $1$.
        \item  We add a special voter $v^\star$ who approves the candidate $c^\star$.
        \item The cost of all the arcs is $1$.
        \item The budget $k$ is $\tilde{k}$.
    \end{itemize}
    
Let $H$ be the delegation graph that we constructed. Note that in this graph the score of $c^\star$ is $0$ and $k$ for every other candidate. 

Next, we prove the correctness. 
 In the forward direction, suppose that $S$ is a solution to $(G,\tilde{k})$. Without loss of generality, we assume that $|S|=\tilde{k}$. If $v\in S$, then we redirect the arc $(v,v')$ to $(v,v^\star)$. Since the cost of arc $(v,v')$ is $1$, $|S|= \tilde{k}$, and $k=\tilde{k}$, the total cost of redirection is at most $k$. Now, we prove that $c^\star$ is the unique winner in the delegation graph $H^\star$ obtained after redirections. For every vertex $v\in S$, we have an arc $(v,v^\star)$ in $H^\star$. Furthermore, since $v$ is the source vertex (vertex with in-degree zero) in the delegation graph, the vote of $c^\star$ is $k$. For every edge $e(=uv) \in E(G)$, $u$ delegates to $u'$, $v$ delegates to $v'$, and $u',v'$ both approve $e$. Since either $u$ or $v$ is in $S$, the vote of candidate $c_e$ decreases by at least $1$, hence, the score in $H^\star$ is at most $k-1$ ($k-4$ votes are due to dummy voters). Thus, $c^\star$ is the unique winner in $H$.

 In the reverse direction, suppose that $H^\star$ is the delegation graph obtained after redirecting arcs within the budget such that $c^\star$ is the unique winner. Let $S$ be the set of arcs that are redirected in $H$ to obtain $H^\star$. We can safely assume that all the arcs are redirected to $v^\star$ because if an arc is redirected to some other voter say $v$, then instead of $v$, redirecting it to $v^\star$ increases the score of $c^\star$ and decrease the score of other voters. We modify $S$ a bit as follows. Suppose that an arc $(d,e)\in S$, and $(u,e) \notin S$, where $u$ is a voter corresponding to the vertex $u\in V(G)$, which is an endpoint of the edge $e\in E(G)$. Then, delete $(d,e)$ from $S$ and  add $(u,e)$ in $S$. Note that the cost of both the arcs is same, so the total cost remain same. Redirection of arc $(d,e)$ decreases the score of the candidate $c_e$ by $1$ and increase the score of $c^\star$ by $1$. The same is also achieved by redirecting the arc $(u,e)$.  Thus, $S$ is still a solution. Since the cost of each redirection is $1$, without loss of generality, we assume that $|S|=k$

 We construct a set $S'\subseteq V(G)$ as follows: if an arc $(u,e)\in S$, add the vertex $u\in V(G)$ in $S'$. Next, we prove that $S'$ is a vertex cover of $G$ of size at most $k$.  Since every redirection, increases the score of $c^\star$ by  $1$, the score of $c^\star$ is $k$. Thus, the score of other candidates is at most $k-1$. Recall that the score of each candidate, except $c^\star$, in $H$ is $k$. Thus, the score of each $c_e$, where $e\in E(G)$ reduced by at least $1$. Thus, for every $e(=uv)\in E(G)$, either $(u,u')$ is redirected or $(v,v')$ is redirected in $S$ due to our modification of the set $S$. Hence, for every $e(=uv)\in E(G)$, either $u\in S'$ or $v \in S'$. Thus, $S'$ is a vertex cover of $G$.

This completes the proof of Theorem~\ref{thm:nph-single-deleg} with constraints 1(a), 2, and 3.


   \begin{figure}[h]
    \centering
    \includegraphics[scale=0.55]{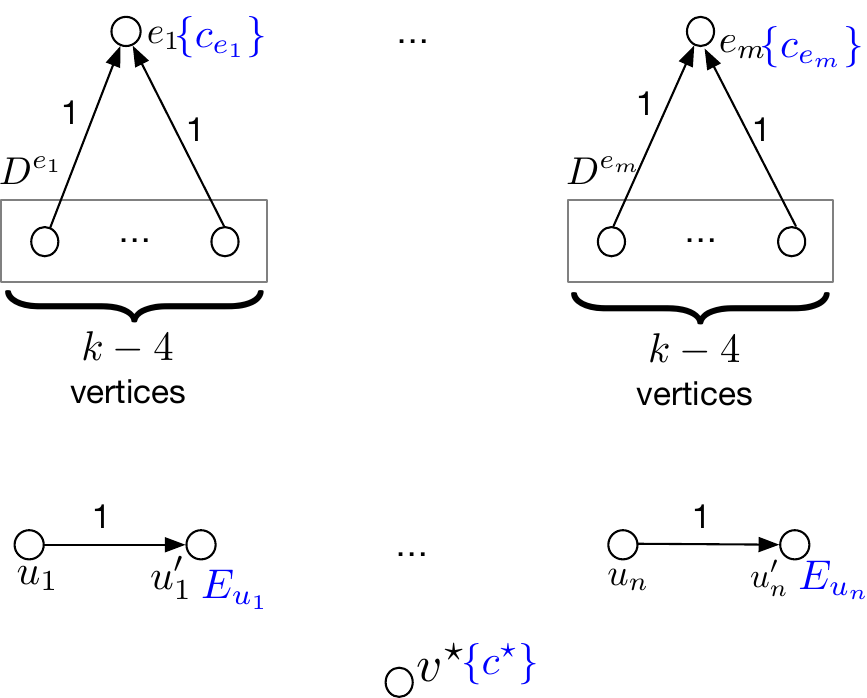}
    \caption{Illustration of {\sf NP}-hardness claimed in Theorem~\ref{thm:nph-single-deleg}. $E_u$ is the set of candidates corresponding to the edges incident to $u$ in $G$. The set of approved candidates by an active voter is written in the {\color{blue}blue} color next to the voter name. The number on an edge is the cost of redirection.}
    \label{fig:thm1-3}
\end{figure}

To prove the result for the constraint 1(b), 2, and 3, we do the same modification as in Theorem~\ref{thm:nph-general}. Note that the only active voters $e$, for every $e\in E(G)$ has large in-degree. All the other vertices have in-degree (out-degree) at most $1$. Thus, to decrease the out-degree of $e$, we add a path on the dummy voters. That is, for every $i\in [k-5], e\in E(G)$, $d^e_i$ delegates to $d^e_{i+1}$ and $d^e_{k-4}$ delegates to the voter $e$. Figure~\ref{fig:thm1-4} illustrates the modification. The cost of arc of type $(u,u')$ is $1$ and for the remaining arcs, it is $k+1$. The correctness is same as earlier. Now, we can only redirect the arcs of type $(u,u')$.

\begin{figure}[ht]
    \centering
    \includegraphics[scale=0.55]{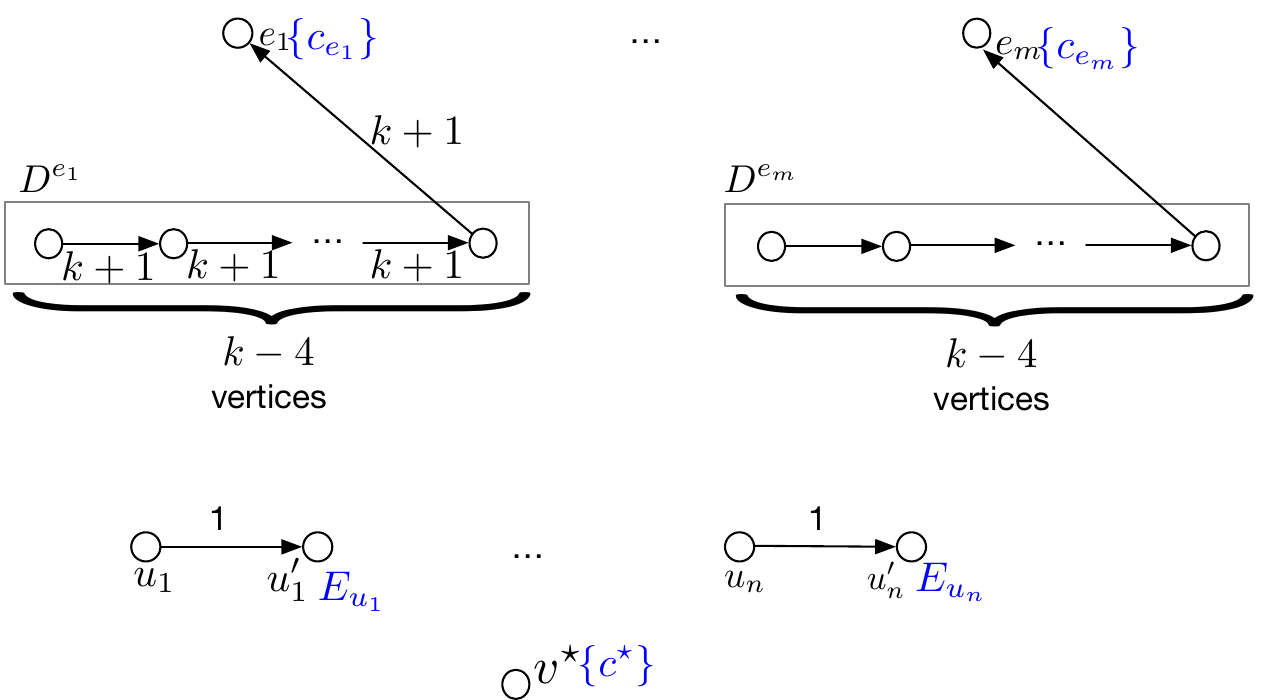}
    \caption{Modifications in Figure~\ref{fig:thm1-4} for {\sf NP}-hardness of Theorem~\ref{thm:nph-general} with constraint 1(b).}
    \label{fig:thm1-4}
\end{figure}
 
Instead of \vc, again if we give the same reduction from {\sc Hitting Set}, we obtain the claimed  {\sf W[2]}-hardness with respect to $k+\redirect$, but the size of approval set is no longer a constant. 
 
\end{proof}

The complexity of \ccra when $\deleg=1$ and $\app=2$ elude us so far. 
\section{Algorithmic Results}
In this section, we design exact, parameterized, and approximation algorithms for \ccra in various settings.

\paragraph{Preprocessing.} We perform the following simple preprocessing step on the given \ccra instance. We partition the active voters into equivalence classes based on their approval sets, i.e., for each subset $C' \subseteq C$ of candidates, let $S_{C'}$ be the set of active voters whose approval set is equal to $C'$. 
Then, for each $C'$ such that $S_{C'}$ is non-empty, we add a \emph{virtual} active voter $v_{C'}$---who does not have an actual vote, i.e., does not add to the resulting scores--- and add arcs from each original active voter $v \in S_{C'}$ to $v_{C'}$ of cost $\infty$. It is easy to see that this results in an equivalent instance of \ccra wherein the number of active voters is upper bounded by the number of subsets $C' \subseteq C$ such that $S_{C'} \neq \emptyset$, which is at most $2^m$. In the following algorithms, we work with the instances of \ccra preprocessed in this manner.

\subsection{Polynomial-Time Algorithm}
In the following theorem, we give a polynomial-time algorithm for \ccra in the single-delegation, single-approval setting.
\begin{theorem}\label{thm:ooly-single-deleg}
    \ccra is polynomial-time solvable when $\deleg = 1$ and $\app = 1$.
\end{theorem}
\begin{proof}    
    Due to preprocessing and $\app = 1$, we know that each active voter votes for a distinct candidate and $t = m$. Let us arbitrarily number the active voters as $v_0, v_1, v_2, \ldots, v_t$ such that 
    each $v_i$ votes for candidate $c_i$, where $c_t = c^\ast$, our preferred candidate. Since $\deleg = 1$, the set of vertices currently voting for $c_i$, by the virtue of the active voter $v_i$ forms a tree, say $T_i$. First, we prove the following lemma that is used as a subroutine in the main algorithm. 

    \begin{lemma} \label{lem:simple-tree-dp}
    There exists a polynomial-time algorithm that returns, for each $1 \le i \le t$, an array $A_i$ of length $|V(T_i)|+1$, such that for each $0 \le j \le |V(T_i)|$, $A_i[j]$ is equal to the minimum-cost of a subset of arcs from $T_i$, whose redirection yields at least $j$ votes for $c^{*}$. One can also compute in polynomial-time, as the corresponding subsets of arcs realizing the minimum cost.
 \end{lemma}    
\begin{proof}
We perform a bottom-up dynamic programming on the $T_i$, which we root at the active voter $v_i$. Thus, every arc $(u, v)$ is directed from a child $u$ to its parent $v$. For each vertex $v$ with parent $w$, we associate two arrays, namely $A_v$ and $B_v$ -- since the root $v_i$ does not have a parent, we only define $B_{v_i}$ for it, which is also the final output of the algorithm. 

Let $\ell_v$ denote the number of nodes in the subtree rooted at $v$ (denoted by $T_v$), then $A_v$ has length $\ell_v+1$ and $B_v$ has length $\ell_v$, these two arrays represent the computational primitives for the dynamic steps in the algorithm. For $0 \le j \le \ell_{v}-1$, $B_v[j]$ is the minimum cost of arcs from $E(T_v)$, whose redirection yields at least $j$ votes from $V(T_v) \setminus \LR{v}$; whereas for $0 \le j \le \ell_v$, $A_v[j]$ is the minimum cost of arcs from $E(T_v) \cup \LR{(v, w)}$, whose redirection yields at least $j$ votes from $V(T_v)$.

\textbf{Leaf case.} If $u$ is a leaf, then $B_u$ is an array of length $1$, with $B_u[0] \coloneqq 0$. 
The correctness follows since the subtree does not contain any arc, which implies that it is not possible to obtain any votes from $V(T_v)$.

\textbf{Computing $A_{v}$ given $B_v$.} Let $v$ be a vertex (other than the root $v_i$) with parent $w$. Suppose we have already correctly computed the array $B_v$ of length $\ell = |V(T_v)|$. Then, we define $A_{v}[\ell] = cost((v, w))$, since redirection of edge $(v, w)$ is the only way to obtain $\ell$ votes. For all other $0 \le j \le \ell-1$, define \\
\begin{equation}
    A_v[j] = \min\LR{B_v[j], cost((v, w))}.    
\end{equation}

, since there are two possibilities to obtain $j$ votes -- either by redirecting a subset of arcs from $E(T_v)$, which is correctly computed in $B_v[j]$; or by redirecting $(v, w)$, which costs $cost((v, w))$ and yields $\ell \ge j$ votes.

\textbf{Computing $B_v$ given the arrays $A$ for the children.} Consider a vertex $v$ with children $u_1, u_2, \ldots, u_r$ (numbered arbitrarily), and for each $1 \le q \le r$, let $e_q = (u_q, v)$ be the edge between $v$ and $q$-th child. Suppose that we have computed the arrays $A_{v_q}$ for every $v_q$. For simplicity, we denote $A_{v_q}$ by $A_q$. We compute the array $B_v$ by iteratively merging the arrays for the children, as follows. First, we merge $A_{\le 1} \coloneqq A_1$ and $A_2$ to compute $B_{\le 2}$ (discussed next). Then, we merge $A_{\le 2}$ with $A_3$ to compute $A_{\le 3}$, and so on. The final output $B_v$ is equal to $A_{\le t}$ computed in this manner.

Now we discuss how to compute $A_{\le y+1}$, given an array $A_{\le y}$ of length $\ell'_y+1$, where $\ell'_y = |\bigcup_{1 \le q \le y} V(T_{u_y})|$, and an array $A_{y+1}$ of length $\ell_{y+1}+1$. This is done as follows. For each $0 \le j \le \ell'_y + \ell_{y+1} + 1$, define \\ 
\begin{equation}
     A_{\le y+1}[j] = \min_{0 \le x \le \min\LR{j, \ell'_{y}}} \LR{A_{\le y}[x]+A_{y+1}[j-x]}
 \end{equation}

Now we discuss the correctness of this computation. We inductively assume that, for each $0 \le x \le \ell'_y$, $A_{\le y}[x]$ is equal to the minimum cost of arcs from $\bigcup_{1 \le q \le y} E(T_{u_q})$ whose redirection yields at least $x$ votes from $\bigcup_{1 \le q \le y} V(T_{u_q})$. In order to obtain $j$ votes from $\bigcup_{1 \le q' \le y+1} E(T_{u_{q'}})$, we must obtain at least $0 \le x \le \min\LR{\ell'_y, j}$ votes from $\bigcup_{1 \le q \le y} V(T_{u_q})$, and remaining at least $j-x$ votes from $V(T_{u_{y+1}})$. Since we consider all possibilities for $x$, the correctness of the computation follows. It is easy to modify the algorithm to also return the corresponding solutions, we omit the details.
\end{proof}

We use the algorithm in Lemma \ref{lem:simple-tree-dp} to compute the arrays $A_i[\cdot]$ for each $0 \le i < t$ (since $v_t$ votes for $c^\ast$, we will not redirect any arcs out of tree $T_t$). We use these arrays to perform another level of dynamic programming, to solve \ccra optimally.
Our DP table is parameterized by a tuple, $(i, k, x)$, where $0 \le i < t$, and $0 \le x, k \le n$, and the corresponding entry $T[i, k, x]$ is equal to the minimum cost of arcs redirected from $T_1 \cup \ldots \cup T_i$ to get at least $k$ votes for $c^\ast$ (i.e., the arcs will be redirected into $v_t$), and leaving at most $x$ votes for the first $i$ candidates, $c_1, \ldots, c_i$.

For the base case, we define $T[0, k, x] = 0$ iff $k = x = 0$ and $T[0, k, x] = +\infty$ otherwise. Suppose we have already computed all the entries of the form $T[i-1, k, x]$. Let $l_i = |V(T_i)|+1$ denote the number of voters in the $i$-th tree. We compute the entries in the $i$-th row using the following recurrence.
\begin{equation}
    T[i,k,x] = \min_{l_i - x \le j \le l_i} \LR{A_{i}[j] + T[i-1,k-j,x]}
\end{equation}

To see the correctness of the computation, fix some $1 \le i < t$. In order to obtain at least $k$ votes from the first $i$ trees, we must obtain some $j$ votes from $T_j$, and remaining at least $k-j$ votes from the first $i-1$ trees. However, we also require that after all such redirections, first $i$ candidates are left with at most $x$ votes. It follows that $j$ must be at least $l_i - x$, and due to Lemma \ref{lem:simple-tree-dp}, $A_i[j]$ contains the minimum cost required to obtain at least $j$ votes from $T_i$. Further, due to inductive hypothesis, the entry $T[i-1, k-j, x]$ contains the minimum cost of arcs from the first $i-1$ trees to obtain at least $k-j$ votes, leaving at most $x$ votes for $c_1, \ldots, c_{i-1}$. Since we take a minimum over all such $(j, k-j)$ combinations, the correctness follows.
The final output of the algorithm is the minimum value $T[t, k, x]$ over all $0 \le x < k \le n$. 
\end{proof}

\subsection{Parameterized Algorithms}
We design the \emph{fixed-parameter tractable} (\fpt) algorithms with respect to the parameters, number of vertices ($n$), the number of active voters (t), and number of candidates ($m$). First, we have the following result.
\begin{theorem}\label{thm:fpt-n}
    \ccra is \fpt\ w.r.t. the parameter $n$.
\end{theorem}
\begin{proof}
    The proof is by enumerating solutions and returning a minimum-cost solution found wherein our preferred candidate $c^{\ast}$ is the unique winner. Note that the number of redirected arcs in any solution is at most $n^2$, and for each redirected arc, there are at most $n$ choices for the new destination. Thus, the number of solutions is upper bounded by $2^{n^2} \cdot n^{n^2} = 2^{O(n^2 \log n)}$. In the single-delegation scenario (i.e., $\deleg = 1$), note that the number of edges in the graph is upper bounded by $n$, and thus the running time of the algorithm improves to $2^{O(n \log n)}$.
\end{proof}

Now we turn to \ccra when $\deleg = 1$ and $\app \ge 1$, and first design an \xp\ algorithm for \ccra parameterized by the number of active voters (Theorem \ref{thm:xp-active}). Subsequently, we show that in the special case of this setting, one can obtain an FPT approximation scheme (FPT-AS) for \ccra parameterized by the number of active voters (Theorem \ref{thm:fpt-as}), i.e., a $(1+\epsilon)$-approximation for the minimum cost that runs in time $f(t, \epsilon) \cdot (m+n)^{O(1)}$ where $t$ is the number of active voters. 

\begin{restatable}{theorem}{xpactive} \label{thm:xp-active}
    \ccra is \xp\ with respect to $t$, the number of active voters when $\deleg = 1$. 
\end{restatable}
\begin{proof}[Proof sketch]
    For each active voter $v_i$, $1 \le i \le t$, let $T_i$ be the corresponding tree and let $l_i = |V(T_i)|$ denote the number of voters in the tree.

    Fix $1 \le i \le t$. We generalize the dynamic programming algorithm from Lemma \ref{lem:simple-tree-dp} to first compute the answers to the subproblems of the following form. Let $S = \LR{v_1, v_2, \ldots, v_{t'}}$ be a multiset of size $1 \le t' \le t$ (with repetitions) such that $v_j \in \LR{1, 2, \ldots, l_i}$. Then, let $T[i, S]$ denote the minimum cost of deleting a subset of arcs $A$ from $T_i$ such that, $A$ can be decomposed into $A_1 \uplus A_2 \uplus \ldots A_{t'}$, where the total number of votes obtained from the set $A_j$ is at least $v_j$. 

    For each non-root vertex $v \in T_i$ with parent $w$, and for each multiset $S = \LR{n_1, n_2, \ldots, n_{t'}}$, we define a subproblem $T_{v}[S]$, denoting the minimum cost of arcs $A \subseteq E(T_v) \cup \LR{(v, w)}$ that can be decomposed as $A_1 \uplus A_2 \ldots \uplus A_{t'}$, where for each $1 \le j \le t'$ the number of votes obtained by redirecting the arcs in $A_j$ is at least $n_j$. For the base case, corresponding to a leaf vertex $v$, we define $T_{v}[S] = w(e)$ iff $S = \LR{1}$, and for all other multisets we define $T_e[S] = \infty$. For the root/active voter $v_i$, the definition is analogous except that the set of arcs $A$ is a subset of $E(T_{v_i})$, since $v_i$ does not have a parent.

    Next, consider a vertex $v$ with children $u_1, u_2, \ldots, u_q$ (numbered arbitrarily), and parent $w$. Let $e_j = (u_j, v)$. Suppose for each $1 \le j \le q$, we have correctly computed the tables $T_{u_j}[S]$ for all children $u_j$ and for all multisets $S$. Let $T_{\le 1}[\cdot] = T_{u_1}[\cdot]$. Fix a particular multiset $S = \LR{n_1, n_2, \ldots, n_{t'}}$. We say that $(S_1, S_2)$ is a valid partition of $S$ iff $S_1 = \LR{a_1, a_2, \ldots, a_{t'}}$ and $S_2 = \LR{b_{1}, b_{2}, \ldots, b_{t'}}$ satisfy for each $1 \le j \le t'$, $n_j = a_{j} + b_{j}$ with $0 \le a_{j}, b_{j} \le n_j$. Then, we compute $T_{\le j}[S]$ using the following recurrence: 
    $T_{\le j}[S] = \min \LR{T_{\le j-1}[S_1] + T_{u_j}[S_2]}$, where the minimum is taken over all valid partitions $(S_1, S_2)$ of the set $S$. Finally, once we have computed the table $T_{\le q}[S]$, we compute the table $T_{v}[S]$ as follows (this step is omitted for the root). $T_{vv'}[\LR{n_{v}}] = w(vv')$ where $n_{v}$ denotes the number of vertices in the subtree rooted at $v$ (including $v$). For all other multisets $S$, we define $T_{vv'}[S] = T_{\le q}[S]$. Finally, the table $T[i, S]$ is equal to the table computed for the root $v_i$. It is easy to see that the running time of the algorithm is $|n_i|^{O(t)}$, and the correctness of this computation can be argued in a similar way to that in Lemma \ref{lem:simple-tree-dp}.

    Suppose that we have computed the tables $T[i, S]$ for each $i$ and each multiset $S$. Now, fix an optimal solution. For each $1 \le i \neq j \le t$, let $n(i, j)$ denote the number of votes redirected by the optimal solution from $T_i$ into $v_j$. We guess the numbers $n(i, j)$ -- note that the number of guesses is upper bounded by $n^{t^2}$. Fix one such guess. For each $1 \le i \le t$, define the multiset $S_i = \LR{n(i, j): 1 \le j \le t, j \neq i}$, and we look up the optimal cost of redirecting a subset of arcs from $T_i$ to get the votes as given in $S_i$ using the entry $T[i, S_i]$. The total cost corresponding to this guess is defined as $\sum_{i = 1}^t T[i, S_i]$. Finally, we return the minimum solution found over all guesses. It follows that the algorithm runs in time $n^{O(t^2)}$ and returns an optimal solution.
\end{proof}

Now we consider the special case of \ccra in the single delegation, multi-approval setting, where our preferred candidate $c^\ast$ appears \emph{separately} from the rest of the candidates. More specifically, we assume that there exists an active voter, say $v_t$, whose approval set is equal to $\LR{c^\ast}$, and the approval sets of other active voters do not contain $c^\ast$ -- note that they may be arbitrary subsets of $C \setminus \LR{c^\ast}$. At an intuitive level, this setting is easier to handle due to the fact that, in an optimal solution, each redirected edge is redirected into $v_t$. We proceed to the following theorem that gives an \fpt\ approximation scheme ({\sf FPT-AS}).

\begin{restatable}{theorem}{fptas} \label{thm:fpt-as}
    \ccra with $\deleg = 1$ and $\app \ge 1$, admits an {\sf FPT-AS} parameterized by $t$ and $\epsilon$ in the special setting, where $c^\ast$ appears separate from the rest of the candidates. That is, in this setting, for any $0 \le \epsilon \le 1$, there exists an algorithm that runs in time $(t/\epsilon)^{O(t)} \cdot (m+n)^{O(1)}$ and returns a solution of cost at most $(1+\epsilon)$ times the optimal redirection cost. 
\end{restatable}
\begin{proof}
    First, by iterating over all edges, we ``guess'' the most expensive edge in the optimal solution. Consider one such guess corresponding to edge $e$ with cost $w$. In the following, we describe the algorithm assuming this guess is correct, i.e., $w$ is indeed the weight of the most expensive edge in the optimal solution. Then, we want to find a solution that only contains edges of cost at most $w$, which implies that $OPT$, the cost of the optimal solution is at least $w$ and at most $n \cdot w$ (recall that $\deleg = 1$, so the total number of edges in $G$ is at most $n$).

    Next, for each active voter $v_i \neq v_1$, let $w_i$ denote the total cost of the edges redirected from the tree $T_i$ into $v_1$ in the optimal solution. Note that $\sum_i w_i = OPT$. Let $w'_i = \max\LR{w_i, \epsilon w/2n}$. Note that $\epsilon w/n \le w'_i \le w n$, and $\sum_{i = 2}^t w'_i \le (1+\epsilon/2) \cdot OPT$, where $OPT$ denotes the cost of an optimal solution. Next, let $p_i = \lceil \log_{1+\epsilon}(w'_i) \rceil$, i.e., $(1+\epsilon/3)^{p_i}$ is the smallest power of $1+\epsilon$ that is at least $w'_i$. Note that for each $2 \le i \le t$, $\log_{1+\epsilon/3}(\epsilon w/2n) \le p_i \le \log_{1+\epsilon/3}(nw)+1$, which implies that all the value of $p_i$ lies in the interval of range $r = \log_{1+\epsilon/3} \lr{\frac{nw}{\epsilon w/2n}} = O(\frac{\log n}{\epsilon^2})$.

    Next, we guess the value of $p_i$ for each $2 \le i \le t$ -- note that the number of guesses is at most $r^{t-1}$. Let $p'_2, p'_3, \ldots, p'_t$ be the guessed values. Then, we use the polynomial-time algorithm in Lemma \ref{lem:simple-tree-dp} to find, for each $2 \le i \le t$, the largest number of votes that can be obtained from the tree $T_i$ using cost at most $(1+\epsilon/3)^{p'_i}$, and redirect such edges to $v_1$. After doing this for all $2 \le i \le t$, we check whether $c^\ast$ is the unique winning candidate. We return the minimum-cost solution found over all guesses for $w$ and all guesses for $p'_i$.  Note that when we correctly guess the values of $p_i$, it holds $w_i \le (1+\epsilon/3)^{p_i}$, thus, the maximum number of votes from $T_i$ that can be obtained using a budget of $(1+\epsilon/3)^{p_i}$ is at least the number of votes that can be obtained using a budget of $w_i$. Thus, in the iteration corresponding to the correct guess, when each guessed value is indeed equal to $p'_i$ , the number of votes obtained for $c^\ast$ is no fewer than that in the optimal solution (since from each tree $T_i$, the number of votes obtained by a solution of cost $(1+\epsilon/3)^{p'_i}$ is no fewer than that of cost $w_i$). Furthermore, the cost of the combined solution thus obtained, is at most $(1+\epsilon/2) \cdot (1+\epsilon/3) \le (1+\epsilon)$ times the cost of an optimal solution.

    Finally, we argue about the running time, which is dominated by the number of guesses for the values of $p_i$. The total number of guesses is at most $\lr{\frac{\log n}{\epsilon}}^{O(t)}$. Now we consider two cases. If $t \le \frac{\log n}{\log \log n}$, then the previous quantity is at most $\frac{1}{\epsilon^{O(t)}} \cdot n^{O(1)}$. Otherwise, if $t > \frac{\log n}{\log \log n}$, then $\lr{\frac{\log n}{\epsilon}}^{O(t)} \le \lr{\frac{t}{\epsilon}}^{O(t)}$. Thus, in either case, we can upper bound the number of guesses, and in turn the running time of the algorithm by $(t/\epsilon)^{O(t)} \cdot (m+n)^{O(1)}$. 

\end{proof}
Recall that our preprocessing step bounds the number of active voters $t$ by the number of distinct approval sets, which is at most $2^m$. Thus, Theorem \ref{thm:xp-active} and Theorem \ref{thm:fpt-as} immediately result in the following corollary.

\begin{sloppypar}
\begin{corollary} \label{cor:xp-m}
    Consider \ccra when $\deleg = 1$ and $\app \ge 1$. In this setting, the problem is \xp\ w.r.t. $m$ (the number of candidates), and admits an FPT-AS parameterized by $m$ and $\epsilon$ in the special case when $c^\ast$ appears separate from the rest of the candidates.
\end{corollary}
\end{sloppypar}

\section{Outlook}

As liquid democracy is gaining more attention in different applications -- as well as more real-world usage, including in high-stakes ones -- it is important to study different aspects of it, including the possibility of external agents controlling and rigging elections performed using liquid democracy.
In this context, here we considered a particular form of election control for liquid democracy by redirecting arcs. 
Below we describe few future research directions that we view as important and promising:
(1) A natural expansion of our work may be to consider not only single-winner elections but also multiwinner elections; and, related, to consider other unraveling functions $\mathcal{R}$ and other voting rules $\mathcal{W}$.
(2) A different generalization of our work is to consider general underlying social graphs, as described in Remark~\ref{remark:arbitrarygraph}.
(3) An immediate future research direction is to consider other forms of election control and bribery, e.g. do not redirect arcs but change the delegation graph in other ways.
(4) Another future research direction would be to study the type of control we study here but from a more practical point of view, e.g., by performing computer-based simulations to estimate the feasibility of successfully controlling a real world election by redirecting liquid democracy arcs.

\bibliographystyle{siam}
\bibliography{aamas24}

\end{document}